\documentclass[conference,a4paper]{IEEEtran}

\newtheorem{theorem}{Theorem}

\newtheorem{example}{Example}
\newtheorem{definition}{Definition}
\usepackage{graphicx,cite}

\usepackage{blindtext, graphicx, amsfonts,
amssymb,multirow}

\ifCLASSINFOpdf

\else

\fi

\usepackage[cmex10]{amsmath}

\usepackage{algorithm}
\usepackage{algorithmic}

\newcommand{\calV}{\mathcal{V}}
\newcommand{\calU}{\mathcal{U}}
\newcommand{\calN}{\mathcal{N}}
\newcommand{\calP}{\mathcal{P}}
\newcommand{\calS}{\mathcal{S}}
\newcommand{\calD}{\mathcal{D}}
\newcommand{\calT}{\mathcal{T}}
\newcommand{\calH}{\mathcal{H}}

\newcommand{\bfV}{\mathbf{V}}

\begin{document}
%
\title{Coded Caching for Networks with the Resolvability Property}

\author{\IEEEauthorblockN{Li Tang and Aditya Ramamoorthy}
\IEEEauthorblockA{Department of Electrical and Computer Engineering\\
Iowa State University\\
Ames, IA 50010\\
Emails:\{litang, adityar\}@iastate.edu}
}


%


\maketitle

\begin{abstract}

Coded caching is a recently proposed technique for dealing with large scale content distribution over the Internet. As in conventional caching, it leverages the presence of local caches at the end users. However, it considers coding in the caches and/or coded transmission from the central server and demonstrates that huge savings in transmission rate are possible when the server and the end users are connected via a single shared link. In this work, we consider a more general topology where there is a layer of relay nodes between the server and the users, e.g., combination networks studied in network coding are an instance of these networks. We propose novel schemes for a class of such networks that satisfy a so-called resolvability property and demonstrate that the performance of our scheme is strictly better than previously proposed schemes.
\end{abstract}


%
\IEEEpeerreviewmaketitle

\section{Introduction}
Caching is a popular technique for facilitating content delivery over the Internet. It exploits local cache memory that is often available at the end users for reducing transmission rates from the central server. In particular, when the users request files from a central server, the system first attempts to satisfy the user demands in part from the local content. Thus, the overall rate of transmission from the server is reduced which in turn reduces overall network congestion. The work of \cite{Ma14} demonstrated that huge rate savings are possible when coding in the caches and coded transmissions from the server to the users are considered. This problem is referred to as coded caching. 

In \cite{Ma14}, the scenario considered was as follows. There is a server that contains $N$ files, a collection of $K$ users that are connected to the central server by a single shared link. Each user also has a local cache of size $M$. The focus is on reducing the rate of transmission on the shared link. There are two distinct phases in the coded caching setting. In the placement phase, the caches of the users are populated. This phase should not depend on the actual user requests, which are assumed to be arbitrary. The scheme in \cite{Ma14} operates by dividing each file into a large number of subfiles; this will be referred to as the subpacketization level.

In this work, we consider the coded caching problem in a more general setting where there is a layer of relay nodes between the server and the users (see \cite{JJ15} for related work). Specifically, the server is connected to a set of relay nodes and the users are connected to certain subsets of the relay nodes. A class of such networks have been studied in network coding and are referred to as ``combination networks" \cite{RWY04}. Specifically, in a combination network there are $h$ relay nodes and the $\binom{h}{r}$ users each of which is connected to a $r$-subset of the the relay nodes. Combination networks were the first example where an unbounded gap between network coding and routing for the case of multicast was shown. While this setting is still far from an arbitrary network connecting the server and the users, it is rich enough to admit several complex strategies that may shed light on coded caching for general networks.

In this work we consider a class of networks that satisfy a so called resolvability property. These networks include combination networks where $r$ divides $h$, but there are many other examples.  We propose a coded caching scheme for these networks and demonstrate its advantages.
\subsection{Main contributions}
\begin{itemize}
\item Our schemes work for any network that satisfies the resolvability property. For a class of combination networks, we demonstrate that our achievable rates are strictly better than those proposed in prior work.
\item The subpacketization level of our scheme is also significantly lower than competing methods. As discussed in Section \ref{sec:prob_form}, the subpacketization level of a given scheme directly correlates with the complexity of implementation.
\end{itemize}

This paper is organized as follows. Section \ref{sec:prob_form} presents the problem formulation and background. In Section \ref{sec:prop_scheme}, we describe our proposed coded caching scheme, Section \ref{sec:perf_anal} presents a performance analysis and comparison and Section \ref{sec:concl} concludes the paper.

\section{Problem Formulation and Background}
\label{sec:prob_form}
In this work we consider a class of networks that contain an intermediate layer of nodes between the main server and the end user nodes. Combination networks are a specific type of such networks and have been studied in some depth in the literature on network coding \cite{RWY04}. However, as explained below, we actually consider a larger class of networks that encompass combination networks.

The networks we consider consist of a server node denoted $\calS$ and $h$ relay nodes, $\Gamma_1, \Gamma_2, \dots, \Gamma_h$ such that the server is connected to each of the relay nodes by a single edge. The set of relay nodes is denoted by $\mathcal H$. Let $[m] = \{1, 2, \dots, m\}$. If $A \subset [h]$ we let $\Gamma_A = \cup_{i\in A} \{\Gamma_i\}$.
There are $K$ users in the system and each user is connected to a subset of $\mathcal H$ of size $r$. Let $\calV \subset \{1, \dots, h\}$ with $|\calV|=r$. For convenience, we will assume that the set $\calV$ is written in ascending order even though the subset structure does not impose any ordering on the elements. Under this condition, we let $\calV[i]$ represent the $i$-th element of $\calV$. For instance, if $\calV =\{1,3\}$, then $\calV[1] = 1$ and $\calV[2] = 3$. Likewise, $\mbox{Inv}-\calV$ will denote the corresponding inverse map, i.e. $\mbox{Inv}-\calV[i] = j$ if $\calV[j] = i$.

Each user is labeled by the subset of relay nodes it is connected to. Thus, $U_{\calV}$ denotes the user that is connected to $\Gamma_{\calV}$. The set of all users is denoted $\calU$ and the set of all subsets that specify the users is denoted $\bfV$, i.e., $\calV \in \bfV$ if $U_{\calV}$ is a user. We consider networks where $\bfV$ satisfies the resolvability property that is defined below.
\begin{definition} {\it Resolvability property.} The set $\bfV$ defined above is said to be resolvable if there exists a partition of $\bfV$ into subsets $\calP_1, \calP_2, \dots, \calP_{\tilde{K}}$ such that 
\begin{itemize}
\item for any $i \in [\tilde{K}]$ if $\calV \in \calP_i$ and $\calV' \in \calP_i$, then $\calV \cap \calV' = \emptyset$, and
\item for any $i \in [\tilde{K}]$, we have $\cup_{\calV : \calV \in \calP_i} \calV = [h]$.
\end{itemize}
The subsets $\calP_i$ are referred to as the parallel classes of $\bfV$.
\end{definition}
Each relay node $\Gamma_i$ is thus connected to a set of users that is denoted by $\calN(\Gamma_i)$. A simple counting argument shows that $|\calN(\Gamma_i)| = Kr/h = \tilde{K}$.

Suppose that $r$ divides $h$ and let $\bfV$ be the set of all subsets of size $r$ of $[h]$. In this case, the network defined above is the combination network \cite{JJ15} with $K=\binom{h}{r}$ users. The fact that this network satisfies the resolvability property is not obvious and follows from a result of \cite{BA75}.
\begin{example}
\label{eg:combnet}
The combination network for the case of $h=4, r=2$ is shown in Fig. \ref{Fig:DirectedGraph} and the corresponding parallel classes are
\begin{align*}
\calP_1 &= \{\{1,2\},\{3,4\}\},\\
\calP_2 &= \{\{1,3\},\{2,4\}\}, \text{~and}\\
\calP_3 &= \{\{1,4\},\{2,3\}\}.
\end{align*}
\end{example}
On the other hand, there are other networks where $|\bfV|$ is strictly smaller than $\binom{h}{r}$.
\begin{example}
\label{eg:gen_resolv_net}
Let $h=9, r=3$ and let $\bfV = \{\{1,2,3\},\{4,5,6\}, \{7,8,9\}, \{1,4,7\}, \{2,5,8\}, \{3,6,9\}\}$. In this case, the parallel classes are
\begin{align*}
\calP_1 &= \{\{1,2,3\},\{4,5,6\}, \{7,8,9\}\}, \text{~and}\\
\calP_2 &= \{\{1,4,7\}, \{2,5,8\}, \{3,6,9\}\}.
\end{align*}
\end{example}
We discuss generalizations of these examples in Section III.

The server $\calS$ contains a library of $N$ files where each file is of size $F$ bits (we will interchangeably refer to $F$ as the subpacketization level). The files are represented by random variables $W_i, i =1, \dots, N$, where $W_i$ is distributed uniformly over the set $[2^F]$. Each user has a cache of size $MF$ bits. There are two distinct phases in the coded caching problem. In the placement phase, the content of the user's caches is populated. This phase should not depend on the file requests of the users. In the delivery phase, user $U_{\calV}$ requests a file denoted $W_{d_{\calV}}$ from the library; the set of all user requests is denoted $\calD = \{W_{d_\calV}: U_{\calV} \text{~is a user}\}$. It can be observed that there are a total of $N^K$ distinct request sets. The server responds by transmitting a certain number of bits that satisfies the demands of all the users. A $(M,R_1,R_2)$ caching system also requires the specification of the following encoding and decoding functions.
\begin{itemize}
\item {\it $K$ caching functions:} $Z_{\calV} = \phi_{\calV}(W_1, \dots, W_N)$ which represents the cache content of user $U_{\calV}$. Here, $\phi_{\calV}: [2^{NF}] \rightarrow [2^{MF}]$.
\item {\it $h N^K$ server to relay encoding functions:} The signal $\psi_{\calS \rightarrow \Gamma_i, \calD}(W_1, \dots, W_N)$ is the encoding function for the edge from $\calS$ to relay node $\Gamma_i$. Here, $\psi_{\calS \rightarrow \Gamma_i, \calD}: [2^{NF}] \rightarrow [2^{R_1 F}]$, so that the rate of transmission on server to relay edges is at most $R_{1}$. The signal on the edge is denoted $X_{\calS \rightarrow \Gamma_i}$.
\item {\it $h \tilde{K} N^K$ relay to user encoding functions:} Let $U_\calV \in \calN(\Gamma_i)$. The signal $\varphi_{\Gamma_i \rightarrow U_\calV, \calD}(\psi_{\calS \rightarrow \Gamma_i, \calD}(W_1, \dots, W_N))$ is the encoding function for the edge $\Gamma_i \rightarrow U_\calV$. Here $\varphi_{\Gamma_i \rightarrow U_\calV, \calD}:[2^{R_{1} F}] \rightarrow [2^{R_{2} F}]$, so that the rate of transmission on the relay to user edges is at most $R_2$. The signal on the corresponding edge is denoted $X_{\Gamma_i \rightarrow U_\calV}$, which is assumed to be defined only if $U_\calV \in \calN(\Gamma_i)$.
\item {\it $K N^K$ decoding functions:} Every user has a decoding function for a specific request set $\calD$, denoted by $\mu_{\calD,U_\calV}(X_{\Gamma_{\calV[1]} \rightarrow  U_\calV}, \dots, X_{\Gamma_{\calV[r]} \rightarrow  U_\calV})$. Here $\mu_{\calD,U_\calV}:[2^{R_{1} F}] \times [2^{R_{2} F}] \times \dots [2^{R_{2} F}] \rightarrow [2^F]$. The decoded file is denoted by $\hat{W}_{\calD,U_\calV}$.
\end{itemize}
For this coded caching system, the probability of error $P_e$ is defined as
$P_e = \max_{\calD} \max_{\calV} P(\hat{W}_{\calD,U_\calV} \neq W_{\calD})$.

\begin{figure}[t]
 \centering
        \includegraphics[scale=0.64]{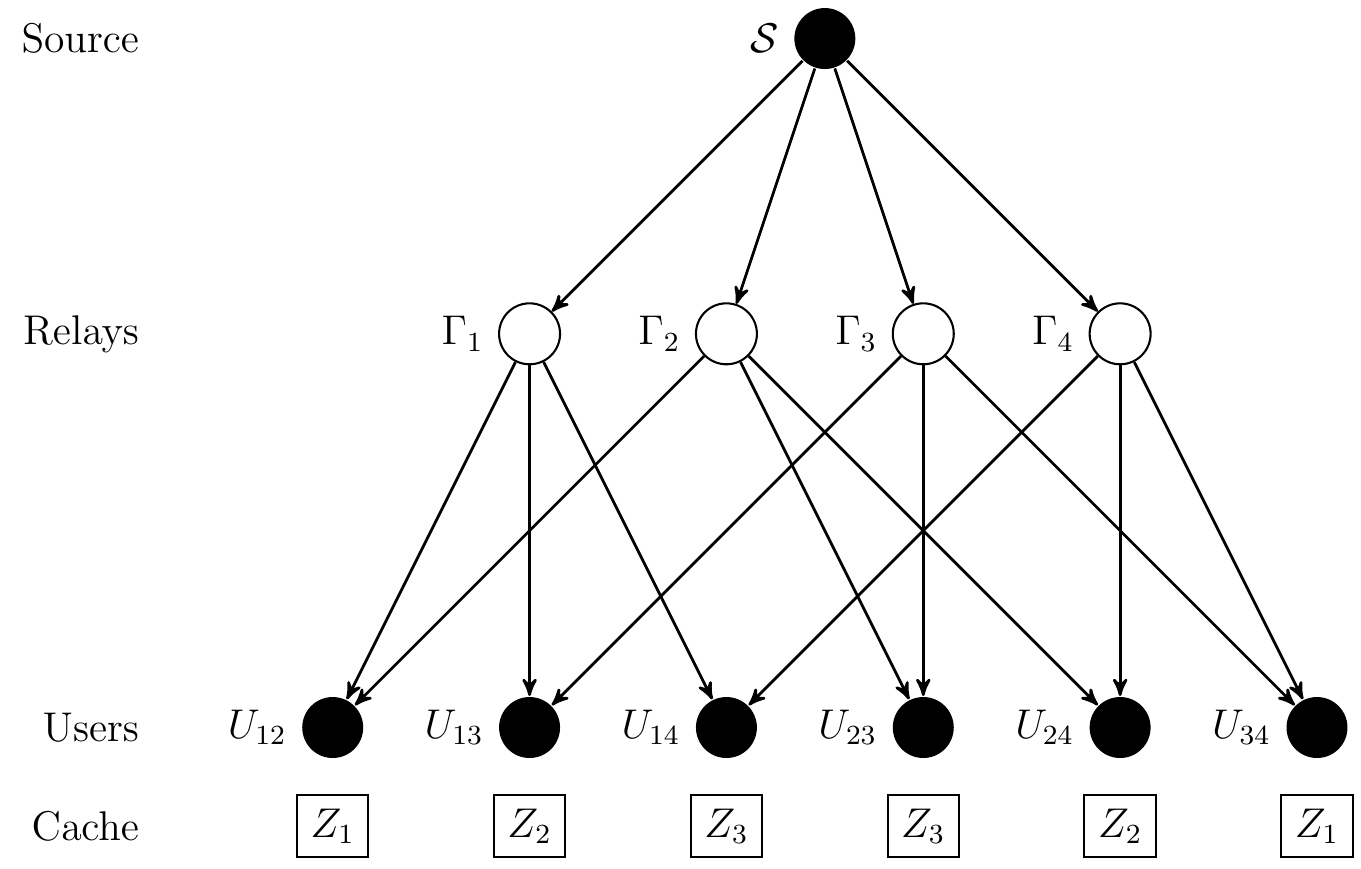}
        \caption{The figure shows a $\binom{4}{2}$ combination network. It also shows the cache placement when $M=2, N=6$. Here, $Z_1 = \cup_{n=1}^6 \{W_{n, 1}^1, W_{n, 1}^2\}, Z_2 = \cup_{n=1}^6 \{W_{n, 2}^1, W_{n, 2}^2\}$ and $Z_3 = \cup_{n=1}^6 \{W_{n, 3}^1, W_{n, 3}^2\}$. It can be observed that each relay node sees the same caching pattern in the users that it is connected to, i.e., the users connected to each $\Gamma_i$ together have $Z_1, Z_2$ and $Z_3$ represented in their caches.}
         \label{Fig:DirectedGraph}
\vspace{-0.2in}
\end{figure}

The triplet $(M,R_1, R_2)$ is said to be achievable if for every $\epsilon > 0$ and every large enough file size $F$ there exists a $(M,R_1,R_2)$ caching scheme with probability of error less than $\epsilon$.  The subpacketization level of a scheme, i.e., $F$ is also an important metric, because it is directly connected to the implementation complexity of the scheme. For example, the original scheme of \cite{Ma14} that operates when there is a shared link between the server and the users. It operates with a subpacketization level $F \approx \binom{K}{KM/N}$ which grows exponentially with $K$. Thus, the scheme of \cite{Ma14} is applicable when the files are very large. In general, lower subpacketization levels for a given rate are preferable.

Prior work has presented two coded caching schemes for combination networks. In the routing scheme, coding is not considered. Each user simply caches a $M/N$ fraction of each file. In the delivery phase, the total number of bits that need to be transmitted is $K(1-M/N)F$. As there are $h$ outgoing edges from $\calS$, we have that $R_1 = \frac{K}{h}(1-\frac{M}{N})$. Moreover, as there are $r$ incoming edges into each user, $R_2 = \frac{1}{r}(1-\frac{M}{N})$. An alternate scheme, called the CM-CNC scheme was presented in \cite{JJ15} for combination networks. This scheme uses a decentralized cache placement phase, where each user randomly caches a $M/N$ fraction of each file. In the delivery phase, the server executes the CM step where it encodes all requested files by a decentralized multicasting caching scheme. Following this, in the CNC step, the server divides each coded signal into $r$ equal-size signals and encodes these $r$ signals by a $(h,r)$ binary MDS code. The $h$ coded signals are transmitted over the $h$ links from the server to relay nodes. The relay nodes forward the signals to the users. Thus, each user receives $r$ coded signals and can recover its demand due to MDS property. 

\section{Proposed Caching Scheme}
\label{sec:prop_scheme}
Consider a network where $\bfV$ satisfies resolvability property and let the parallel classes be $\mathcal P_1,\cdots, \mathcal P_{\tilde{K}}$. It is evident that each user belongs to exactly one parallel class. Let $\Delta(\calV)$ indicate the parallel class that a user belongs to, i.e., $\Delta(\calV) = j$ if user $U_\calV$ belong to $\calP_j$. Now, recall that $\calN(\Gamma_i)$ is the set of users $U_\calV$ such that $i \in \calV$. By the resolvability property it has to be the case that each user in $\calN(\Gamma_i)$ belongs to a different parallel class. In fact, it can be observed that
\begin{align}
\label{eq:subnet_classes}
\{\Delta(\calV) : U_\calV \in \calN(\Gamma_i) \} = [h], \text{~for all $i$.}
\end{align}
This implies that each relay node ``sees" exactly the same set of parallel classes represented in the users that it is connected to. This observation inspires our placement phase in the caching scheme.  We populate the user caches based on the parallel class that a given user belongs to. Loosely speaking, it turns out that we can design a symmetric uncoded placement such that the overall cache content seen by every relay node is the same.

\begin{algorithm}[t]
\caption{Coded Caching in Networks Satisfying Resolvability Property}
\label{Alg:CachingCombination}
\begin{algorithmic}
      \STATE {\bfseries 1. procedure:} PLACEMENT PHASE 

      \FOR{$i=1$ {\bfseries to} $N$}
      \STATE Partition $W_i$ into $(W_{n,\mathcal T}^l:\mathcal T\subset[\tilde{K}], |\mathcal T|=t, l\in [r])$
      \ENDFOR
      \FOR{$\calV \in \mathbf{V}$}
      \STATE $U_{\mathcal V}$ caches $W_{n,\mathcal T}^l$ if $\Delta(\calV) \in \mathcal T$ for  $l\in [r]$ and $n \in [N]$.
      \ENDFOR
   \STATE {\bfseries end procedure}
   \STATE

   \STATE {\bfseries 2.procedure:} DELIVERY PHASE
   \FOR{$i=1$ {\bfseries to} $h$}
    \STATE $\Theta \leftarrow \{C: C\subset [\tilde{K}], |C|=t+1\}$
    \STATE Source sends $\{\oplus_{\{\calV:\Delta(\calV)\in C\}} W_{d_{\mathcal V}, {C\setminus \{\Delta(\calV)\} }}^{\mbox{Inv}-\calV[i]}: i\in \mathcal V, C\in \Theta \}$ to $\Gamma_i$
    \FOR{$j=1$ {\bfseries to} $\tilde{K}$}
     \STATE  $\Gamma_i$ forwards $\{\oplus_{\{\calV:\Delta(\calV)\in C\}} W_{d_{\mathcal V}, {C\setminus \{\Delta(\calV)\} }}^{\mbox{Inv}-\calV[i]}: i \in \mathcal V, C\in \Theta, \Delta(\calV)\in C\}$ to $U_{\mathcal V}$
    \ENDFOR
   \ENDFOR
      \STATE {\bfseries end procedure}

\end{algorithmic}
\end{algorithm}

Our proposed placement and delivery phase schemes are formally specified in Algorithm \ref{Alg:CachingCombination} and are discussed below.
Assume each user has a storage capacity of $M\in\{0,\frac{N}{\tilde{K}},\frac{2N}{\tilde{K}},\cdots,N\}$ files, and let $t=\frac{\tilde{K}M}{N}=\frac{KrM}{hN}$. The users can be partitioned into $\tilde{K}$ groups $\mathcal G_i$ where $\mathcal G_i=\{U_\mathcal V:\mathcal V\in \mathcal P_i\}$.

In placement phase, each file $W_n$ is split into $r\binom{\tilde{K}}{t}$ non-overlapping subfiles of equal size that are labeled as
$$
W_n=(W_{n,\mathcal T}^l:\mathcal T\subset[\tilde{K}], |\mathcal T|=t, l\in [r]).
$$
Thus, the subpacketization mechanism is such that each subfile has a superscript in $[r]$ in addition to the subset-based subscript that was introduced in the work of \cite{Ma14}.

Subfile $W_{n,\mathcal T}^l$ is placed in the cache of the users in $\mathcal G_i$ if $i\in \mathcal T$. Equivalently, $W_{n,\mathcal T}^l$ is stored in user $U_{\mathcal V}$ if $\Delta(\calV)\in \mathcal T$. Thus, each user caches a total of $Nr\binom{\tilde{K}-1}{t-1}$ subfiles, and each subfile has size $\frac{F}{r\binom{\tilde{K}}{t}}$. This requires
$$
Nr\binom{\tilde{K}-1}{t-1}\frac{F}{r\binom{\tilde{K}}{t}}=F\frac{Nt}{\tilde{K}}=MF
$$
bits, demonstrating that our scheme uses only $MF$ bits of cache memory at each user.

\begin{example}
\label{eg:combnet_with_M_2}
Consider the combination network in Example \ref{eg:combnet} with $M=2$, $N=6$ and $K=6$. In the placement phase, each file is partitioned into six subfiles $W_{n, i}^l$, $i=1,2,3$, $l=1,2$.
The cache placement is as follows.
\begin{align*}
\mathcal G_1=\{U_{12},U_{34}\} \text{~cache~}& W_{n, 1}^1, W_{n, 1}^2;\\
\mathcal G_2=\{U_{13},U_{24}\} \text{~cache~}& W_{n, 2}^1, W_{n, 2}^2; \text{~and}\\
\mathcal G_3=\{U_{14},U_{23}\} \text{~cache~}& W_{n, 3}^1, W_{n, 3}^2.
\end{align*}
\end{example}

Note that by eq. (\ref{eq:subnet_classes}), we have that each relay node is connected to a user from each parallel class. Our placement scheme depends on the parallel class that user belongs to. Thus, it ensures that the overall distribution of the cache content seen by each relay node is the same. This can be seen in Fig. \ref{Fig:DirectedGraph} for the example considered above. We note here that the routing scheme ({\it cf.} Section \ref{sec:prob_form}) is also applicable for this placement.


Now, we briefly outline the main idea of our achievable scheme. Our file parts are of the form $W_{n,\calT}^{j}$, where for a given $\calT$, $j \in [r]$. Note that each user is also connected to $r$ different relay nodes in $\calH$. Our proposed scheme is such that each user recovers a missing file part with a certain superscript from one of the relay nodes it is connected to. In particular, we convey enough information from the server to the relay nodes such that each relay node uses the scheme proposed in \cite{Ma14} for one set of superscripts. Crucially, the symmetrization afforded by the placement scheme, allows each relay node to operate in this manner. 


\begin{theorem}
	\label{Theo:Achievable}
	Consider a network satisfying resolvability property with $h$ relay nodes and $K$ users such that each user is connected to $r$ relay nodes. Suppose that the $N$ files in the server and each user has cache of size $M\in\{0,\frac{Nh}{Kr},\frac{2Nh}{Kr},\cdots,N\}$. Then, the following rate pair $(R_1, R_2)$ is achievable.
\begin{align}
	R_1 &=\min{\left\{\frac{K(1-\frac{M}{N})}{h(1+\frac{KrM}{hN})}, \frac{N}{r}(1-\frac{M}{N})\right\}}, \label{eq:r1_rate}\\
	R_2 &=\frac{1-\frac{M}{N}}{r}. \nonumber
\end{align}
For general $0\le M\le N$, the lower convex envelope of these points is achievable.
\end{theorem}	
\begin{proof}	
Let the set of user requests be denoted $\calD=\{W_{d_\calV}:U_\calV\text{~is~a~user}\}$. 
For each relay node $\Gamma_i$, we focus on the users connected to it $\calN(\Gamma_i)$ and a subset $C \subset \{\Delta(\calV): U_\calV \in \calN(\Gamma_i)\}$ where $|C|=t+1$. For each subset $C$, the server transmits
\begin{align}
\label{eq:server_tx}
	\oplus_{\{\calV:\Delta(\calV)\in C\}} W_{d_{\mathcal V}, {C\setminus \{\Delta(\calV)\}}}^{\mbox{Inv}-\calV[i]}
\end{align}
to the relay node $\Gamma_i$ ($\oplus$ denotes bitwise XOR) and $\Gamma_i$ forwards it into users $U_{\mathcal V}$ where $\Delta(\calV)\in C$.


We now argue that each user can recover its requested file. Evidently, a user $U_{\calV}$ is missing subfiles of the form $W_{d_{\mathcal V}, \mathcal{T}}^{j}$ where $\Delta(\calV) \notin \calT$. If $U_\calV$ is connected to $\Gamma_i$, it can recover the following set of subfiles using the transmissions from $\Gamma_i$.
$$
\{W_{d_{\mathcal V}, \mathcal{T}}^{\mbox{Inv}-\calV[i]}: \mathcal{T}\subset [\tilde{K}]\setminus \{\Delta(\calV)\}, |\mathcal{T}|=t\}.
$$
This is because the transmission in eq. (\ref{eq:server_tx}) is such that $U_\calV$ caches all subfiles that are involved in the XOR except the one that is interested in. This implies it can decode its missing subfile.
In addition, $U_\calV$ is also connected to $r$ relay nodes so that $\cup_{i \in \calV} \{\mbox{Inv}-\calV[i]\} = [r]$, i.e., it can recover all its missing subfiles.


Next, we determine $R_1$ and $R_2$. Each of coded subfiles results in $\frac{F}{r\binom{\tilde{K}}{t}}$ bits being sent over the link from source $\mathcal S$ to $\Gamma_i$. Since the number of subsets $C$ is $\binom{\tilde{K}}{t+1}$, the total number of bits sent from $\mathcal S$ to $\Gamma_i$ is $\binom{\tilde{K}}{t+1}\frac{F}{r\binom{\tilde{K}}{t}}=F\frac{\tilde{K}(1-\frac{M}{N})}{r(1+\frac{\tilde{K}M}{N})}$, and hence $R_1=\frac{\tilde{K}(1-\frac{M}{N})}{r(1+\frac{\tilde{K}M}{N})}$. Next, note that each coded subfile is forwarded to $|C|=t+1$ users. Thus, each user receives $\frac{|C|\binom{\tilde{K}}{t+1}}{\tilde{K}}$ coded subfiles so that the total number of bits sent from a relay node to user is  $\frac{|C|\binom{\tilde{K}}{t+1}}{\tilde{K}}\times \frac{F}{r\binom{\tilde{K}}{t}}=\frac{|C|F(\tilde{K}-t)}{r\tilde{K}(t+1)}=F\frac{1-\frac{M}{N}}{r}$. Hence $R_2=\frac{1-\frac{M}{N}}{r}$. 

Thus, the triplet $(M, R_1, R_2)=(M,\frac{\tilde{K}(1-\frac{M}{N})}{r(1+\frac{\tilde{K}M}{N})}, \frac{1-\frac{M}{N}}{r})$ is achievable for $M\in\{0,\frac{Nh}{Kr},\frac{2Nh}{Kr},\cdots,N\}$ . Points for general values of $M$ can be obtained by memory sharing between triplets of this form.
%

If $\tilde{K}>N$, it is clear that some users connecting to a given relay node $\Gamma_i$ request the same file. In this case, the routing scheme can attain $R_1=\frac{N}{r}(1-\frac{M}{N})$, which is better than the proposed scheme if $M\le 1-\frac{N}{\tilde{K}}$. This explains the second term within the minimum in the RHS of eq. (\ref{eq:r1_rate}).
\end{proof}
We illustrate our achievable scheme by considering the setup in Example \ref{eg:combnet_with_M_2}.
\begin{example}
Assume that user $U_{\mathcal V}, \calV \subset \{1, \dots, 4\}, |\calV| = 2$, requires file $W_{d_{\mathcal V}}$. The users connected to $\Gamma_1$ correspond to subsets $\{1,2\}, \{1,3\}$ and $\{1,4\}$ so that $\mbox{Inv}-\calV[1] = 1$ for all of them. Thus, the users recover missing subfiles with superscript of $1$ from $\Gamma_1$. In particular, the transmissions are as follows.
\begin{align*}
&S\to \Gamma_1: W_{d_{12},2}^1\oplus W_{d_{13},1}^1, W_{d_{12},3}^1\oplus W_{d_{14},1}^1, W_{d_{13},3}^1\oplus W_{d_{14},2}^1,\\
&\Gamma_1\to U_{12} :  W_{d_{12},2}^1\oplus W_{d_{13},1}^1, W_{d_{12},3}^1\oplus W_{d_{14},1}^1,\\
&\Gamma_1\to U_{13} :  W_{d_{12},2}^1\oplus W_{d_{13},1}^1, W_{d_{13},3}^1\oplus W_{d_{14},2}^1, \text{~and}\\
&\Gamma_1\to U_{14} :  W_{d_{12},3}^1\oplus W_{d_{14},1}^1, W_{d_{13},3}^1\oplus W_{d_{14},2}^1.
\end{align*}

The users connected to $\Gamma_2$ correspond to subsets $\{1,2\}, \{2,3\}$ and $\{2,4\}$ in which case $\mbox{Inv}-\{1,2\}[2] = 2$ while  $\mbox{Inv}-\{2,3\}[2] = 1$,  $\mbox{Inv}-\{2,4\}[2] = 1$.
Thus, user $U_{12}$ recovers missing subfiles with superscript $2$ from $\Gamma_2$ while users $U_{23}$ and $U_{24}$ recover missing subfiles with superscript $1$.
The specific transmissions are given below.
\begin{align*}
&S\to \Gamma_2: W_{d_{23},1}^1\oplus W_{d_{12},3}^2, W_{d_{12},2}^2\oplus W_{d_{24},1}^1, W_{d_{23},2}^1\oplus W_{d_{24},3}^1,\\
&\Gamma_2\to U_{12}: W_{d_{23},1}^1\oplus W_{d_{12},3}^2, W_{d_{12},2}^2\oplus W_{d_{24},1}^1,\\
&\Gamma_2\to U_{23}: W_{d_{23},1}^1\oplus W_{d_{12},3}^2, W_{d_{23},2}^1\oplus W_{d_{24},3}^1, \text{~and}\\
&\Gamma_2\to U_{24}: W_{d_{12},2}^2\oplus W_{d_{24},1}^1, W_{d_{23},2}^1\oplus W_{d_{24},3}^1.
\end{align*}
In a similar manner, the other transmissions can be determined and it can be verified that the demands of the users are satisfied and $R_1=\frac{1}{2}$, $R_2=\frac{1}{3}$.
%
%
\end{example}
It is important to note that the resolvability property is key to our proposed scheme. For example, if $r$ does not divide $h$, the combination network does not have the resolvability property. In this case, it can be shown that a symmetric uncoded placement is impossible. We demonstrate this by means of the example below.
\begin{example}
	Consider the combination network with $h=3,r=2$, and $\bfV=\{\{1,2\},\{1,3\},\{2,3\}\}$. Here $2$ does not divide $3$ and it is easy to check that it does not satisfy the resolvability property. Next, we argue that a symmetric uncoded placement is impossible by contradiction. Assume there exists a symmetric uncoded placement and suppose  $U_{12}$ caches $Z_1$, $U_{13}$ caches $Z_2$. By the hypothesis, $\Gamma_1$ and $\Gamma_2$ have to see the same cache content.  Since $\calN(\Gamma_1)=\{U_{12},U_{13}\}$ and $\calN(\Gamma_2)=\{U_{12},U_{23}\}$, $U_{23}$ has to cache $Z_2$. As a result, since $\calN(\Gamma_3)=\{U_{13},U_{23}\}$, $\Gamma_3$ sees $Z_2$ and $Z_2$, which are different from the cache content seen by  $\Gamma_1$ and $\Gamma_2$. This is a contradiction.
	
\end{example}

We emphasize that a large class of networks satisfy the resolvability property. For instance, if $r$ divides $h$, \cite{BA75} shows that the set of all $\binom{h}{r}$ $r$-subsets of an $h$-set can be partitioned into disjoint parallel classes $\mathcal P_i$, $i=1,2,\cdots,\binom{h-1}{r-1}$. More generally, one can consider resolvable designs \cite{Stinson} which are set systems that satisfy the resolvability property. Such designs include affine planes which correspond to networks where for prime $q$, we have $h = q^2$ and the $r = q$; the set $\mathbf{V}$ is given by the specification of the affine plane. Furthermore, one can obtain resolvable designs from affine geometry over $\mathbb{F}_q$ that will correspond to networks with $h = q^m$ and $r = q^d$.

\section{Performance Analysis}
\label{sec:perf_anal}
We now compare the performance of our proposed scheme with the CM-CNC scheme \cite{JJ15} and the routing scheme. For a given value of $M$ we compare the achievable $R_1, R_2$ pairs of the different schemes. Furthermore, we also compare the required subpacketization levels of the different schemes, as it directly impacts the complexity of implementation of a given scheme. Table \ref{Table:Compare} summarizes the comparison. We note here that the rate of the CM-CNC scheme is derived in \cite{JJ15} for a decentralized placement. The rate in Table \ref{Table:Compare} corresponds to a derivation of the corresponding rate for a centralized placement is lower than the one for the decentralized placement.
\begin{figure}[t]
	\centering
	\includegraphics[scale=0.6]{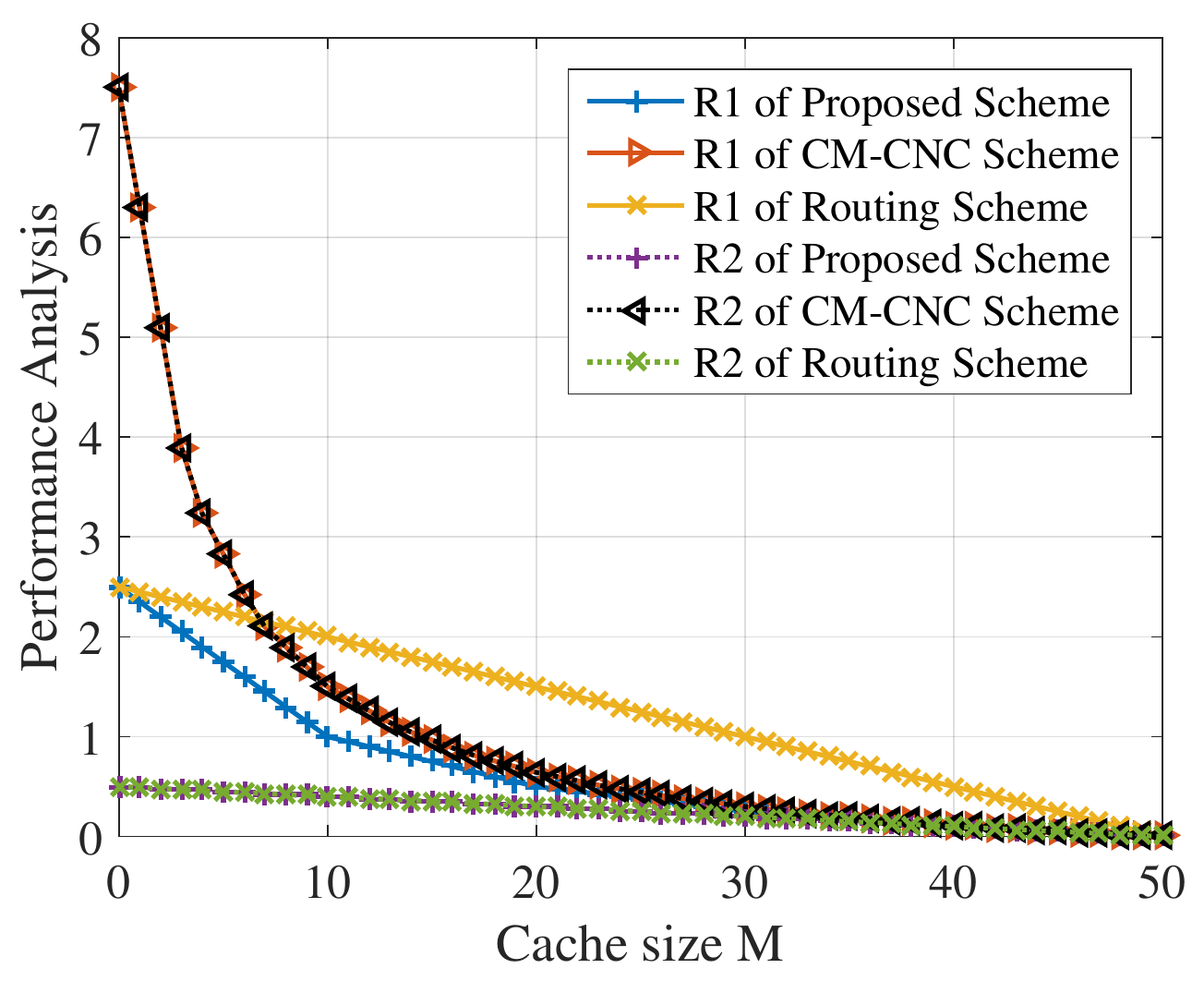}
	\caption{Performance comparison of the different schemes for a $\binom{6}{2}$ combination network with $K=15$, $\tilde{K}=5$ and $N=50$.}
	\label{Fig:R1}
\vspace{-0.2in}
\end{figure}

%
%

The following conclusions can be drawn. Let $(R_1^*, R_2^*)$ and $F^*$ denote the rates and subpacketization level of our proposed scheme.
\begin{align*}
\frac{R_1^{*}}{R_1^{CM-CNC}} &= \frac{\frac{1}{K}+\frac{M}{N}}{\frac{1}{\tilde{K}}+\frac{M}{N}} < 1, \text{~and}\\
\frac{R_2^{*}}{R_2^{CM-CNC}}&=\frac{1}{K}+\frac{M}{N}\\
             &\le \frac{N-\frac{N}{\tilde{K}}}{N}+\frac{1}{K}< 1.
\end{align*}
This implies that our scheme is unbounded better in both rate metrics.
Next,
\begin{align*}
\frac{F^{*}}{F^{CM-CNC}} \approx \exp \left\{(K(1 - \frac{r}{h}) H_{e}(\frac{M}{N}) \right\}
\end{align*}
where $H_{e}(\cdot)$ represents the binary entropy function in nats. Thus, the subpacketization level of our scheme is exponentially smaller than the scheme of \cite{JJ15}.

For a $\binom{6}{2}$ combination network with parameters  $K=15$, $\tilde{K}=5$, $N=50$, we plot the performance of the different schemes in Fig. \ref{Fig:R1}. Fig. \ref{Fig:R1} compares $R_1$ and $R_2$ of three schemes. It can be observed that for $R_1$, the proposed scheme is best for all cache size $M$. At the same time, we can see that $R_2$ of routing scheme and the proposed scheme are identical but significantly better than that of CM-CNC scheme.

It is to be noted that the scheme of \cite{JJ15} operates via a decentralized placement phase where the users cache random subsets of the bits of each file. Our proposed scheme is evidently a centralized scheme and part of the gain can be attributed to the ability to choose the cache content carefully. Nevertheless, the symmetrization of the cache content with respect to the relay nodes is a novel aspect of our work.

\begin{table}[t]\large
\centering
\begin{tabular}{|c|c|c|c|}
\hline
\hline
    & Routing & CM-CNC & New Scheme\\
    \hline
$F$ & $r\binom{\tilde{K}}{\frac{\tilde{K}M}{N}}$  & $r\binom{K}{\frac{KM}{N}}$ & $r\binom{\tilde{K}}{\frac{\tilde{K}M}{N}}$\\
\hline
$R_1$ & $\frac{K}{h}(1-\frac{M}{N})$ & $\frac{K(1-\frac{M}{N})}{r(1+\frac{KM}{N})}$ & $\frac{\tilde{K}(1-\frac{M}{N})}{r(1+\frac{\tilde{K}M}{N})}$\\
\hline
$R_2$ & $\frac{1}{r}(1-\frac{M}{N})$ & $\frac{K(1-\frac{M}{N})}{r(1+\frac{KM}{N})}$ & $\frac{1}{r}(1-\frac{M}{N})$\\
\hline
\hline
\end{tabular}
\caption{Comparison of three schemes}
	\label{Table:Compare}
\vspace{-0.2in}
\end{table}

\section{Conclusions}
\label{sec:concl}
In this work, we proposed a coding caching scheme for networks that satisfy the resolvability property. This family of networks includes a class of combination networks as a special case. The rate required by our scheme for transmission over the server-to-relay edges and over the relay-to-user edges is strictly lesser than that proposed in prior work. In addition, the subpacketization level of our scheme is also significantly lower than prior work. The generalization to networks that do not satisfy the resolvability property and to networks with arbitrary topologies is an interesting direction for future work.

\bibliographystyle{IEEE-unsorted}
\bibliography{refs}

\begin{thebibliography}{1}

\bibitem{Ma14}
M.~A. Maddah-Ali and U.~Niesen,
\newblock ``Fundamental limits of caching,''
\newblock {\em {IEEE} Trans. Info. Theory}, vol.~60, pp.~2856--2867, May 2014.

\bibitem{JJ15}
M.~Ji, M.~F. Wong, A.~M. Tulino, J.~Llorca, G.~Caire, M.~Effros, and
  M.~Langberg,
\newblock ``On the fundamental limits of caching in combination networks,''
\newblock in {\em IEEE 16th International Workshop on Signal Processing
  Advances in Wireless Communications (SPAWC)}, pp. 695--699, June 2015.

\bibitem{RWY04}
C.~K. Ngai and R.~W. Yeung,
\newblock ``Network coding gain of combination networks,''
\newblock in {\em Proc. of IEEE Inform. Theory and Workshop (ITW)}, 2004.

\bibitem{BA75}
Z.~Baranyai,
\newblock ``On the factorization of the complete uniform hypergraph,''
\newblock in {\em Infinite and finite sets (Colloq., Keszthely, 1973; dedicated
  to P. Erdos on his 60th birthday)}, vol.~1, pp. 91--108, 1975.

\bibitem{Stinson}
D.~R. Stinson,
\newblock {\em Combinatorial Designs: Construction and Analysis},
\newblock Springer, 2003.

\end{thebibliography}

\end{document}